\documentclass[11pt]{article}

\usepackage[margin=1in]{geometry}
\usepackage{amsmath,amssymb,amsfonts,bm,mathtools}
\usepackage{graphicx}
\usepackage{booktabs}
\usepackage{array}
\usepackage{tabularx}
\usepackage{multirow}
\usepackage{xcolor}
\usepackage[round,authoryear]{natbib}
\usepackage{hyperref}
\usepackage{enumitem}
\usepackage{comment}
\usepackage{caption}
\usepackage{subcaption}
\usepackage{float}
\usepackage{algorithm}
\usepackage{algpseudocode}
\usepackage{microtype}
\usepackage{amsthm}

\usepackage{amsmath}

\hypersetup{colorlinks=true, linkcolor=blue!55!black, citecolor=blue!55!black, urlcolor=blue!55!black}

\definecolor{mainblue}{RGB}{24,70,130}
\definecolor{softgray}{RGB}{86,92,104}

\newcommand{\R}{\mathbb{R}}

\newcommand{\tr}{\operatorname{tr}}

\newcommand{\argmin}{\operatorname*{arg\,min}}
\newcommand{\argmax}{\operatorname*{arg\,max}}
\newcommand{\sigmoid}{\mathrm{s}}

\newcommand{\cX}{\mathcal{X}}
\newcommand{\cQ}{\mathcal{Q}}

\newcommand{\cS}{\mathcal{S}}

\newtheorem{proposition}{Proposition}

\title{\bfseries Forecast-Ensemble-Based Active Binary-Threshold Query Design for Interval Data Assimilation}
\author{Wataru Hashimoto, Kazumune Hashimoto, and Haru Kuroki}
\date{}
\begin{document}
\maketitle

\begin{abstract}
Data assimilation estimates the evolving state of a dynamical system by combining model forecasts with observations. While many conventional methods assume point-valued measurements, practical sensing systems may instead provide coarse information such as binary, ordinal, inequality, or interval-valued reports. Interval Data Assimilation (IDA) provides a principled framework for assimilating such information, but assumes that the observations to be assimilated are specified in advance. In many sensing settings, however, both where to query and what inequality to ask can be chosen, while only a limited number of queries can be issued. This raises a new observation-design problem: how should informative inequality queries be selected from the forecast uncertainty before their responses are known? We propose \emph{Active Query Design for Interval Data Assimilation} (AQD-IDA), which uses the forecast ensemble to design binary threshold queries prior to the IDA update. AQD-IDA treats both the observation target and the threshold defining the inequality as design variables, thereby allowing the assimilation system to determine not only where to observe but also what question to ask. We develop query-selection criteria that account for forecast uncertainty, redundancy among queries, and anticipated reduction in posterior uncertainty. Experiments on Lorenz--96 and a spherical quasi-geostrophic model show that adaptive query design improves assimilation accuracy under a fixed binary-query budget, with joint observation-target--threshold selection outperforming fixed or separately designed queries. These results demonstrate that actively designing the information supplied to data assimilation can substantially increase the value of coarse observations.
\end{abstract}


\section{Introduction}

\begin{figure}[t]
    \centering
    \includegraphics[width=\linewidth]{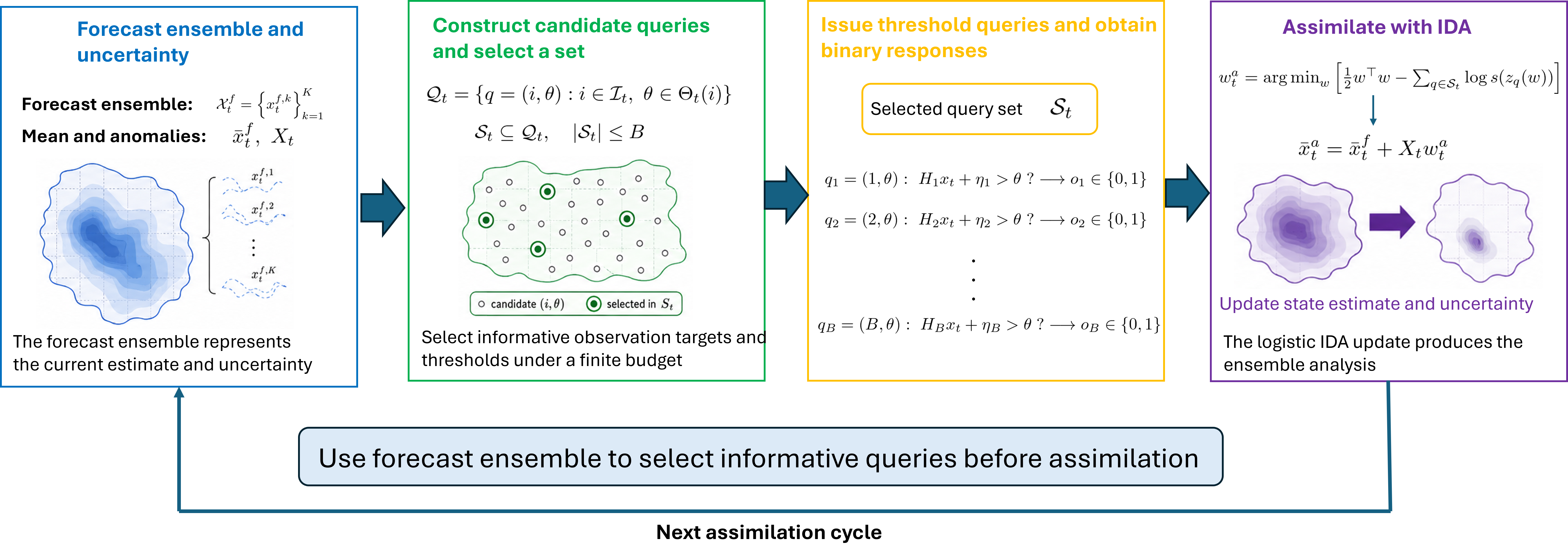}
    \caption{Conceptual workflow of AQD-IDA. The forecast ensemble is used to select
a set of location--threshold queries $q=(i,\theta)$, where $i$ denotes
the observation target and $\theta$ the threshold, under query budget
$B$. The resulting binary responses are assimilated through IDA to
produce the analysis ensemble for the next cycle.}
    \label{fig:concept}
\end{figure}

Data assimilation combines numerical model forecasts with observational information to estimate the evolving state of a dynamical system~\citep{evensen1994sequential,hunt2007letkf}. It is a core methodology in numerical weather prediction, oceanography, hydrology, and many other areas involving high-dimensional nonlinear systems. In a standard formulation, an observation is modeled as a point-valued measurement
\begin{equation}
y_t = H_t x_t + \varepsilon_t,
\qquad \varepsilon_t \sim \mathcal{N}(0,R_t),
\end{equation}
where $x_t\in\R^n$ is the state, $H_t$ is the observation operator, and $R_t$ is the observation error covariance. This leads to a Gaussian likelihood and a quadratic observation cost. Many ensemble data-assimilation methods are built around this point-observation viewpoint, although recent work has also considered more general non-Gaussian ensemble updates~\citep{anderson2022qceff1,anderson2023qceff2}.

In many practical sensing systems, however, point-valued observations are not the natural data type. Low-cost sensors, threshold detectors, image-based classifiers, citizen-science reports, and human observations often provide coarse or thresholded information rather than precise scalar measurements~\citep{leduc2026ida,shah2018semiqualitative,shah2019seaice,srivastava2025binary}. For example, an observing system may report whether a quantity exceeds a threshold or belongs to a coarse category such as low, medium, or high. Such observations are not naturally represented as noisy point values; they are qualitative, inequality, or interval-valued statements. Related data-assimilation studies have treated inequality constraints, truncated distributions, semi-qualitative observations, and categorical or binary data~\citep{thacker2007inequality,lauvernet2009truncated,shah2018semiqualitative,shah2019seaice,srivastava2025binary}. These works show that coarse or constrained observations can be useful in data assimilation and motivate assimilation frameworks beyond Gaussian point-valued observations.

Interval Data Assimilation (IDA) was recently proposed as a unified Bayesian framework~\citep{leduc2026ida} in which observations are interpreted as intervals. A point observation can be represented as a limiting zero-width interval, an inequality observation as an interval extending in one direction, and a finite interval observation as lower and upper bounds. Logistic likelihoods allow these observation types to be represented as smooth inequality likelihoods, yielding a tractable optimization problem in the forecast-ensemble subspace for linear observation operators. In this sense, IDA provides a principled framework for assimilating interval-valued information.

The remaining question, which is the focus of this paper, is how interval-valued information should be acquired before assimilation. In the original IDA formulation, interval observations to be assimilated are assumed to be available in advance. This setting is natural when interval or inequality observations are already determined by the observing system, for example through a detection limit, a known error range, or a physical constraint. In such cases, the resulting observations can simply be assimilated by IDA. In many sensing settings, however, the assimilation system may have several possible observations to acquire, involving different sensing sources, locations, observed variables, observation directions, thresholds, or interval bounds. Acquiring all candidates may be impractical because of human response effort, communication cost, or the computational cost of assimilating many observations. This paper therefore considers the problem of actively selecting a limited number of interval-valued observations under a finite query budget.

This query-design problem is nontrivial because different candidate observations can have very different effects on the IDA update. A query at a location or along a direction with small forecast uncertainty may provide little new information, and a query that is highly redundant with already selected observations may have limited value. For binary threshold queries of the form ``is $H x_t>\theta$?'', the threshold also matters: a threshold far outside the forecast distribution gives an almost deterministic response, whereas a threshold that separates an uncertain part of the forecast ensemble can provide a useful inequality constraint.

This problem is related to existing work on observation design, including targeted observations~\citep{lorenz1998optimal,bishop2001adaptive,majumdar2016review}, observation-impact and forecast-sensitivity methods~\citep{langland2004impact,lien2018efso}, information-based data selection~\citep{migliorini2013information}, sensor placement~\citep{ferrarin2021coastal,castro2020bilevel,mons2017sensor,joshi2009sensor}, and maximum-entropy or submodular sensor selection~\citep{shewry1987maximum,krause2008nearoptimal,krause2014}. 
These studies mainly address where to observe or which existing observations to use. In contrast, the present problem concerns both where to query and what inequality to ask. Specifically, both the observation target and the threshold defining the inequality query are treated as design variables.

Motivated by this gap, this paper proposes
\emph{Active Query Design for Interval Data Assimilation} (AQD-IDA).
At each assimilation cycle, AQD-IDA uses the forecast ensemble to select
a small set of informative binary threshold queries under a finite query
budget. In the main setting considered in this paper, each query is of the
form ``is $H x_t>\theta$?'', which specifies both a scalar quantity
$H x_t$ and a threshold $\theta$. The response is represented as an
inequality observation and assimilated using the IDA update.

We develop AQD-IDA in two steps. We first consider adaptive threshold
selection at fixed observation targets. We then consider joint
observation-target--threshold selection under a finite query budget.
Figure~\ref{fig:concept} illustrates the resulting workflow. Starting from
a forecast ensemble, AQD-IDA selects informative observation targets and
thresholds, issues the corresponding threshold queries, and assimilates
the returned binary responses as inequality observations to produce an
updated analysis ensemble.

The paper is organized as follows. Section~\ref{subsec:ida_summary} reviews the necessary data-assimilation and IDA ingredients. Section~\ref{sec:proposed} develops the proposed AQD-IDA framework. Section~\ref{sec:experiments} evaluates the proposed methods on Lorenz--96 and SHQG benchmarks. Section~\ref{sec:related-work} discusses related work and positions AQD-IDA relative to adaptive observing, sensor selection, and binary-response design. Finally, Section~\ref{sec:conclusion} concludes the paper.

The main contributions are as follows.
\begin{itemize}[leftmargin=1.25em]

\item We formulate a forecast-ensemble-based active query-design problem for
interval-valued observations under a finite query budget, where informative
binary threshold queries are selected before assimilation and the resulting
responses are assimilated through IDA.

\item We develop practical query-selection criteria for adaptive threshold
selection and joint observation-target--threshold selection, using the
forecast ensemble to account for uncertainty, informativeness, and redundancy
among candidate queries.

\item We evaluate the proposed methods on Lorenz--96 and SHQG benchmarks,
demonstrating the benefits of adaptive query design under fixed binary-query
budgets and examining its computational cost and sensitivity to key
assimilation parameters.

\end{itemize}

\section{Review of Interval Data Assimilation}
\label{subsec:ida_summary}

We briefly review Interval Data Assimilation (IDA)~\citep{leduc2026ida}, which serves as the assimilation backbone of this paper. 
IDA extends ensemble data assimilation beyond classical point-valued observations by interpreting observations as probabilistic equality, inequality, or interval statements about the state. 
Using smooth logistic likelihoods, these different observation types are handled in a unified Bayesian objective. 
Given a forecast ensemble, IDA represents the analysis state in the forecast ensemble subspace, solves a likelihood-based optimization problem for the ensemble coefficients, and uses a local Laplace approximation to construct the analysis ensemble. 

At assimilation cycle $t$, let
\begin{equation}
    \mathcal{X}_{t}^{f}
    =
    \left\{
    x_{t}^{f,1},
    \ldots,
    x_{t}^{f,K}
    \right\},
    \qquad
    x_{t}^{f,k}\in\mathbb{R}^{n},
\end{equation}
be the forecast ensemble. Its mean and normalized perturbation matrix are
\begin{align}
    \bar{x}_{t}^{f}
    &=
    \frac{1}{K}
    \sum_{k=1}^{K}x_{t}^{f,k},
    \label{eq:ida_forecast_mean}
    \\
    X_{t}
    &=
    \frac{1}{\sqrt{K-1}}
    \begin{bmatrix}
        x_{t}^{f,1}-\bar{x}_{t}^{f}
        &
        \cdots
        &
        x_{t}^{f,K}-\bar{x}_{t}^{f}
    \end{bmatrix}.
    \label{eq:ida_forecast_perturbation}
\end{align}
The analysis state is represented as
\begin{equation}
    x_{t}
    =
    \bar{x}_{t}^{f}
    +
    X_{t}w,
    \qquad
    w\in\mathbb{R}^{K}.
    \label{eq:ida_ensemble_subspace}
\end{equation}
Here, $w$ is the ensemble coefficient vector, and $I_K$ denotes the $K\times K$ identity matrix. The prior $w\sim\mathcal{N}(0,I_K)$ gives the background penalty
\begin{equation}
    J_{b}(w)
    =
    \frac{1}{2}w^{\top}w.
    \label{eq:ida_background_cost}
\end{equation}
This term discourages corrections that move the analysis too far from the forecast ensemble.

We next describe the observation term. Consider a scalar observation operator $H_i\in\mathbb{R}^{1\times n}$, a threshold $\theta_i\in\mathbb{R}$, and logistic observation noise $\eta_i$ with scale parameter $\sigma_i>0$. Let
\begin{equation}
    s(z)=\frac{1}{1+\exp(-z)},
    \qquad
    \mathbb{P}(\eta_i\le u)=s\left(\frac{u}{\sigma_i}\right).
\end{equation}
Instead of observing the noisy scalar value $H_i x_t+\eta_i$ directly, suppose that the available observation is a binary outcome $o_i\in\{0,1\}$ indicating whether this noisy value exceeds the threshold. For a positive response, the likelihood is
\begin{equation}
    \mathbb{P}(o_i=1\mid x_t)
    =
    \mathbb{P}(H_i x_t+\eta_i>\theta_i\mid x_t)
    =
    s\left(
    \frac{H_{i}x_{t}-\theta_{i}}{\sigma_{i}}
    \right).
    \label{eq:ida_logistic_likelihood_greater}
\end{equation}
Similarly, for a negative response, the likelihood is represented as
\begin{equation}
    \mathbb{P}(o_i=0\mid x_t)
    =
    \mathbb{P}(H_i x_t+\eta_i\le\theta_i\mid x_t)
    =
    s\left(
    \frac{\theta_{i}-H_{i}x_{t}}{\sigma_{i}}
    \right).
    \label{eq:ida_logistic_likelihood_less}
\end{equation}
Thus, upper and lower inequality observations are handled by changing the sign of the logistic argument. Let $b_{i}=1$ for a positive response corresponding to $H_{i}x_{t}>\theta_{i}$ and $b_{i}=-1$ for a negative response corresponding to $H_{i}x_{t}\le\theta_{i}$. Substituting $x_{t}=\bar{x}_{t}^{f}+X_{t}w$, define
\begin{equation}
    z_{i}(w)
    =
    \frac{
    H_{i}\left(\bar{x}_{t}^{f}+X_{t}w\right)-\theta_{i}
    }{
    b_{i}\sigma_{i}
    }.
    \label{eq:ida_logistic_argument}
\end{equation}
Equivalently, with
\begin{equation}
    Y_{i}
    =
    \frac{H_{i}X_{t}}{b_{i}\sigma_{i}},
    \qquad
    d_{i}
    =
    \frac{\theta_{i}-H_{i}\bar{x}_{t}^{f}}{b_{i}\sigma_{i}},
    \label{eq:ida_y_d_definitions}
\end{equation}
we have $z_{i}(w)=Y_{i}w-d_{i}$.

Stacking the $m_t$ logistic inequality terms, define
\begin{equation}
    Y
    =
    \begin{bmatrix}
        Y_{1} \\
        \vdots \\
        Y_{m_t}
    \end{bmatrix}
    \in \mathbb{R}^{m_t\times K},
    \qquad
    d
    =
    \begin{bmatrix}
        d_{1} \\
        \vdots \\
        d_{m_t}
    \end{bmatrix}
    \in \mathbb{R}^{m_t}.
    \label{eq:ida_stacked_y_d}
\end{equation}
Then
\begin{equation}
    z(w)
    =
    \begin{bmatrix}
        z_{1}(w) \\
        \vdots \\
        z_{m_t}(w)
    \end{bmatrix}
    =
    Yw-d.
    \label{eq:ida_stacked_logistic_argument}
\end{equation}
Equality observations can also be included by treating one equality as a pair of lower and upper inequality observations with the same threshold. More general interval observations are handled in the same way using the corresponding lower and upper bounds.

Given $m_{t}$ observations, the IDA analysis coefficient is obtained by solving
\begin{equation}
    w_{t}^{a}
    =
    \arg\min_{w}J_{t}(w),
    \label{eq:ida_analysis_coefficient}
\end{equation}
where
\begin{equation}
    J_{t}(w)
    =
    \frac{1}{2}w^{\top}w
    -
    \sum_{i=1}^{m_{t}}
    \log s(z_{i}(w)).
    \label{eq:ida_objective}
\end{equation}
The first term is the background penalty, and the second term is the negative log-likelihood of the interval observations. Let $s_{i}=s(z_{i}(w))$ and $s=(s_{1},\ldots,s_{m_{t}})^{\top}$, and let $\mathbf{1}\in\mathbb{R}^{m_t}$ denote the vector whose entries are all one. The gradient and Hessian are
\begin{align}
    \nabla J_{t}(w)
    &=
    w
    -
    Y^{\top}(\mathbf{1}-s),
    \label{eq:ida_gradient}
    \\
    \nabla^{2}J_{t}(w)
    &=
    I_K
    +
    Y^{\top}D(w)Y,
    \label{eq:ida_hessian}
\end{align}
where
\begin{equation}
    D(w)
    =
    \operatorname{diag}
    \left(
        s_{1}(1-s_{1}),
        \ldots,
        s_{m_{t}}(1-s_{m_{t}})
    \right).
    \label{eq:ida_curvature_matrix}
\end{equation}
The diagonal factor $s_i(1-s_i)$ is the scalar curvature of the logistic inequality loss. For a single inequality likelihood with argument
\begin{equation}
    z_i=\frac{H_i x_t-\theta_i}{b_i\sigma_i},
\end{equation}
the corresponding Hessian contribution in state space is proportional to
\begin{equation}
    \frac{1}{\sigma_i^2}s_i(1-s_i)H_i^\top H_i.
    \label{eq:ida_scalar_curvature_weight}
\end{equation}
Thus, the local curvature is largest when $z_i=0$, that is, when the threshold is close to the queried scalar value $H_i x_t$. This property will be used in the proposed framework to quantify the informativeness of candidate observations.

The same Hessian expression also gives the well-posedness of the IDA update. Since $D(w)\succeq0$, we have $\nabla^2J_t(w)\succeq I_K$. Hence $J_t$ is $1$-strongly convex in $w$. Moreover, the background term $\frac12 w^\top w$ makes $J_t(w)\to\infty$ as $\|w\|\to\infty$. Therefore, for any finite set of equality, inequality, or interval-valued observations represented by the logistic likelihood terms above, the IDA objective has a unique global minimizer $w_t^a$. 

The analysis mean is
\begin{equation}
    \bar{x}_{t}^{a}
    =
    \bar{x}_{t}^{f}
    +
    X_{t}w_{t}^{a}.
    \label{eq:ida_analysis_mean}
\end{equation}
To construct the analysis ensemble, IDA uses a Laplace approximation around $w_{t}^{a}$. The analysis covariance in ensemble coefficient space is approximated by
\begin{equation}
    P_{w,t}^{a}
    =
    \left[
    \nabla^{2}J_{t}(w)
    \big|_{w=w_{t}^{a}}
    \right]^{-1}
    =
    \left[
    I_K
    +
    Y^{\top}D(w_{t}^{a})Y
    \right]^{-1}.
    \label{eq:ida_laplace_covariance}
\end{equation}
A square root of $P_{w,t}^{a}$ is then used to transform the forecast perturbations and generate the analysis ensemble. The resulting analysis ensemble is propagated by the forecast model to start the next assimilation cycle.

For clarity, one IDA cycle can be summarized as follows:
\begin{enumerate}
    \item Compute $\bar{x}_{t}^{f}$ and $X_{t}$ from the forecast ensemble.
    \item Represent the analysis state as $x_{t}=\bar{x}_{t}^{f}+X_{t}w$.
    \item Convert the available equality, inequality, or interval-valued observations into logistic likelihood terms.
    \item Minimize $J_{t}(w)$ to obtain the analysis coefficient $w_{t}^{a}$.
    \item Compute the analysis mean $\bar{x}_{t}^{a}=\bar{x}_{t}^{f}+X_{t}w_{t}^{a}$.
    \item Approximate the analysis covariance by the inverse Hessian at $w_{t}^{a}$.
    \item Generate the analysis ensemble and propagate it to the next assimilation cycle.
\end{enumerate}
\subsection{Motivation for active query design}

The formulation above concerns the assimilation of interval or
inequality observations once they are available. When observation
sources and reporting formats are fixed in advance, the observations to
be assimilated are predetermined. For example, sensor detection limits,
known error ranges, or physical constraints may directly produce
interval or inequality observations that can be assimilated by IDA
~\citep{leduc2026ida,thacker2007inequality,lauvernet2009truncated,
shah2018semiqualitative}.

In many applications, however, it may be impractical to acquire all available observations. Human and citizen reports require effort from respondents~\citep{see2019citizenscience,assumpcao2018citizen}. Sensor networks may also be constrained by communication bandwidth or energy consumption~\citep{rawat2014wireless}. Acquiring, processing, and assimilating many
observations also increases computational and data-handling costs,
motivating observation thinning and selection in data-assimilation
systems~\citep{ochotta2005adaptive}. Under a finite query budget, the
assimilation system must therefore decide which observations to request.

For a threshold-based observation, a query is defined by two choices:
the observation target---such as a sensor, location, or variable---and
the threshold. The query asks whether the quantity at the selected target
exceeds the selected threshold. A positive response provides a lower
bound on the quantity, whereas a negative response provides an upper
bound. Each target--threshold pair therefore defines a distinct candidate
query whose response yields a one-sided interval observation.

The threshold can serve as a practical design variable in several sensing architectures. In adaptive
one-bit quantization, a sensor compares its continuous measurement with
a threshold set before acquisition and returns only the binary comparison
result to reduce sensing and communication costs~\citep{fang2016adaptive}. If the threshold is programmable, the
acquisition system can set it directly. Alternatively, when several
binary sensors have different fixed thresholds, selecting a sensor also
selects its associated threshold. Both architectures reduce observation
acquisition to choosing among the available target--threshold pairs.

Human reporting provides another setting in which both the reporting location and the reference threshold can be selected. A reporter may be unable to estimate a water level accurately but may be able to determine whether floodwater has exceeded a specified visible reference, such as a mark on a bridge pier or the height of a roadside curb~\citep{see2019citizenscience,assumpcao2018citizen}. When several reporting locations or reference levels are available, the query specifies both where the report is requested and which reference level is used. The response then provides an upper or lower bound on the water level.

Threshold queries may also be preferable when complete measurements
cannot be shared because of privacy, confidentiality, or institutional
restrictions. For example, an organization may report whether a specified
threshold has been exceeded without disclosing the full measurement or
detailed sensor information~\citep{idan2020prshare}. Such responses do not
by themselves guarantee privacy, particularly when adaptive queries are
repeated, but they can reduce the amount of raw information that must be
disclosed.

The benefit of selecting rather than fixing the query is that the
informativeness of a target--threshold pair depends on the current
forecast ensemble. A target with little forecast uncertainty is unlikely
to provide much new information. Even at an uncertain target, a threshold
far outside the forecast distribution produces an almost certain
response and is therefore of limited value. By contrast, a threshold
that separates plausible forecast values can provide a more informative
one-sided constraint
~\citep{srivastava2025binary,fang2016adaptive}. Multiple queries may also
provide redundant information when they probe similar directions of
forecast uncertainty
~\citep{majumdar2016review,krause2008nearoptimal}. Because the forecast
distribution evolves between assimilation cycles, the most informative
targets and thresholds may change over time. These considerations
motivate selecting observation targets and thresholds from the current
forecast ensemble before assimilation.

\section{Proposed Framework: Active Query Design for IDA}
\label{sec:proposed}

In this section, we propose Active Query Design for Interval Data Assimilation
(AQD-IDA). The method is designed for settings where several candidate queries
are available before assimilation, but only a limited number of them can be
issued under a finite query budget. AQD-IDA uses the forecast ensemble to select
informative queries. In the main setting considered here, each query is a binary
threshold query that asks whether a chosen state component or observation
variable exceeds a specified threshold. The resulting binary response is
represented as an inequality observation and is then assimilated using IDA.

We first consider adaptive threshold selection at fixed observation targets.
We then consider joint observation-target--threshold selection under a finite
query budget. The supporting derivations are presented alongside the
corresponding query-selection criteria, and the submodularity property of the
fixed log-det surrogate is stated together with the greedy selection rule.
\subsection{Adaptive threshold selection at fixed observation targets}

Suppose the observation location $i$ is fixed and the sensing platform can issue a binary threshold query
\begin{equation}
    q_{i,\theta}:\quad \text{``Is } x_{i,t}+\eta_{i,t}>\theta \text{?''},
\end{equation}
where $\eta_{i,t}$ is logistic observation noise with scale $\sigma$. The binary outcome is
\begin{equation}
    o_{i,t}(\theta)=\mathbf{1}\{x_{i,t}+\eta_{i,t}>\theta\}\in\{0,1\}.
\end{equation}
After observation, the outcome is represented in IDA by the corresponding logistic inequality likelihood. A positive response contributes
\begin{equation}
    \mathbb{P}(o_{i,t}=1\mid x_{i,t})
    =
    \sigmoid\!\left(\frac{x_{i,t}-\theta}{\sigma}\right),
\end{equation}
whereas a negative response contributes
\begin{equation}
    \mathbb{P}(o_{i,t}=0\mid x_{i,t})
    =
    \sigmoid\!\left(\frac{\theta-x_{i,t}}{\sigma}\right).
\end{equation}
Thus, a positive response is treated as a soft lower inequality, while a negative response is treated as a soft upper inequality. The question is how to choose $\theta$ before observing $o_{i,t}$.

\paragraph{Entropy-based threshold selection.}
We first introduce a simple criterion that selects a threshold by making the two possible binary responses as balanced as possible. The predictive probability of outcome $1$ is approximated by
\begin{equation}
    p_i(\theta)
    = \frac{1}{K}\sum_{k=1}^{K}
    \sigmoid\!\left(\frac{x_{i,t}^{f,k}-\theta}{\sigma}\right),
\end{equation}
where $x_{i,t}^{f,k}$ denotes the $i$-th component of the $k$-th forecast ensemble member at assimilation cycle $t$.
The binary entropy is
\begin{equation}
    \mathcal{H}_i(\theta)
    = -p_i(\theta)\log p_i(\theta)
    - (1-p_i(\theta))\log(1-p_i(\theta)).
\end{equation}
The entropy-based rule selects
\begin{equation}
    \theta_i^\star = \argmax_{\theta\in\Theta_t(i)} \mathcal{H}_i(\theta),
\end{equation}
where $\Theta_t(i)$ is the candidate threshold set for target $i$ at cycle $t$. Since binary entropy is maximized at $p_i(\theta)=1/2$, this rule chooses a threshold for which the two possible responses are approximately balanced. When the observation noise is small, this threshold is close to the predictive median of the forecast ensemble at location $i$.

\paragraph{Curvature-based threshold selection.}
We next introduce a criterion based on the local curvature of the
logistic IDA likelihood. As reviewed in
Section~\ref{subsec:ida_summary}, a logistic inequality likelihood
contributes the scalar curvature weight
$s(z)(1-s(z))/\sigma^2$ to the IDA Hessian. This weight is largest
when the threshold is close to the scalar value being evaluated.

Because the true queried value is unknown before acquisition, we
evaluate this curvature over the forecast ensemble. For component $i$
and candidate threshold $\theta$, define the ensemble-averaged
curvature
\begin{equation}
    C_i(\theta)
    =
    \frac{1}{K}\sum_{k=1}^{K}
    \frac{
    \sigmoid(z_{i,k})
    \bigl(1-\sigmoid(z_{i,k})\bigr)
    }{\sigma^2},
    \qquad
    z_{i,k}
    =
    \frac{x_{i,t}^{f,k}-\theta}{\sigma}.
\end{equation}
The curvature-based rule selects
\begin{equation}
    \theta_i^\star
    =
    \argmax_{\theta\in\Theta_t(i)} C_i(\theta).
\end{equation}
This rule favors thresholds lying in regions of the forecast
distribution where the logistic likelihood has high average
curvature. Such thresholds make a larger local contribution to the
IDA Hessian and are therefore expected to provide more informative
inequality observations under the local Laplace approximation.

\subsection{AQD-IDA: joint observation-target--threshold query design}
\label{subsec:joint_query_design}

The preceding subsection adapts the threshold while keeping the observation
target fixed. AQD-IDA extends the design problem by selecting both the target
and the threshold. A query is written as $q=(i,\theta)$, where $i$ identifies a
scalar observation target and $\theta$ is the threshold. Depending on the
application, $i$ may represent a state component, a spatial location, a
physical variable and vertical level, or, more generally, a scalar observation
operator $H_i$. At assimilation cycle $t$, let $\mathcal I_t$ denote the set of
available observation targets and let $\Theta_t(i)$ be the candidate threshold
set for target $i$. The candidate query set is
\begin{equation}
    \cQ_t
    =
    \{q=(i,\theta): i\in\mathcal I_t,\ \theta\in\Theta_t(i)\}.
    \label{eq:candidate_query_set}
\end{equation}
Given a budget $B$, the design problem is to choose
\begin{equation}
    \cS_t\subseteq\cQ_t,
    \qquad
    |\cS_t|\le B.
    \label{eq:query_budget}
\end{equation}
Each selected query asks whether the noisy scalar quantity
$H_i x_t+\eta_q$ exceeds $\theta$, where $\eta_q$ is an independent
logistic-noise realization associated with query $q$. Thus, when multiple
thresholds are selected for the same observation target, they represent
repeated independent binary acquisitions rather than multiple thresholdings
of a single noisy scalar measurement. The resulting binary responses are
assimilated through the corresponding logistic inequality likelihoods used
by IDA.

\paragraph{Outcome-conditioned expected-gain criterion.}
A direct criterion is the expected reduction in total state variance. With the
normalization in \eqref{eq:ida_forecast_perturbation}, the forecast covariance
is
\begin{equation}
    P_{x,t}^f=X_tX_t^\top.
    \label{eq:forecast_covariance}
\end{equation}
Let $P_{x,t}^a(q,o)$ denote the approximate analysis covariance that would
result if query $q$ were issued and the response were $o\in\{0,1\}$. The
forecast ensemble gives the predictive probability of a positive response,
\begin{equation}
    p_t(q)
    :=
    \mathbb P(o=1\mid q,\cX_t^f)
    \approx
    \frac{1}{K}\sum_{k=1}^{K}
    \sigmoid\!\left(
        \frac{H_i x_t^{f,k}-\theta}{\sigma_q}
    \right),
    \label{eq:predictive_positive_probability}
\end{equation}
where $\sigma_q$ is the logistic noise scale. The expected analysis total variance is the probability-weighted average of the two possible outcome-conditioned analysis variances. We define the expected-gain score as the reduction from the forecast total variance to this expected analysis total variance:
\begin{align}
    \mathcal U_t^{\mathrm{full}}(q)
    ={}&
    \tr(P_{x,t}^f)
    -p_t(q)\tr\!\left(P_{x,t}^a(q,1)\right)
    \nonumber\\
    &-\{1-p_t(q)\}\tr\!\left(P_{x,t}^a(q,0)\right).
    \label{eq:utility}
\end{align}
To evaluate the outcome-conditioned covariances, consider a candidate query
$q=(i,\theta)$ and a hypothetical response $o\in\{0,1\}$. Define
\begin{align}
    z_q(w)
    &:=
    \frac{H_i(\bar{x}_t^f+X_tw)-\theta}{\sigma_q},
    \label{eq:query_argument_in_w}
    \\
    \ell_q(w;o)
    &:=
    -o\log \sigmoid(z_q(w))
    -(1-o)\log\{1-\sigmoid(z_q(w))\}.
    \label{eq:query_negative_loglikelihood}
\end{align}
For query scoring, define the single-query hypothetical IDA objective
\begin{equation}
    J_t^{q,o}(w)
    :=
    \frac{1}{2}w^\top w+\ell_q(w;o),
    \qquad
    w_t^a(q,o):=\argmin_w J_t^{q,o}(w).
    \label{eq:outcome_conditioned_objective}
\end{equation}
A Laplace approximation at the corresponding hypothetical MAP point gives the
local coefficient covariance
\begin{equation}
    P_{w,t}^a(q,o)
    :=
    \left[
        \nabla_w^2J_t^{q,o}(w)
        \big|_{w=w_t^a(q,o)}
    \right]^{-1}.
    \label{eq:outcome_laplace_covariance}
\end{equation}
The corresponding state-space covariance in the forecast ensemble subspace is
\begin{equation}
    P_{x,t}^a(q,o)
    :=
    X_tP_{w,t}^a(q,o)X_t^\top.
    \label{eq:outcome_state_covariance}
\end{equation}
Substituting \eqref{eq:outcome_state_covariance} into \eqref{eq:utility}
gives the outcome-conditioned expected-gain score for each candidate query.

\paragraph{Rank-one expected-gain proxy.}
The outcome-conditioned criterion requires two hypothetical IDA optimizations for every
candidate query, together with two Hessian-based covariance evaluations. This
becomes expensive when the candidate-query set is large. We therefore derive
a less expensive local approximation of the same trace-reduction objective.
Recall that the IDA state is parameterized as
\begin{equation}
    x(w)=\bar{x}_t^f+X_t w.
    \label{eq:query_state_parameterization}
\end{equation}
For a query $q=(i,\theta)$, define
\begin{align}
    z_{q,k}
    &:=\frac{H_i x_t^{f,k}-\theta}{\sigma_q},
    \label{eq:query_standardized_argument}\\
    \omega_q
    &:=\frac{1}{K}\sum_{k=1}^{K}
    \sigmoid(z_{q,k})\{1-\sigmoid(z_{q,k})\},
    \label{eq:query_curvature}\\
    a_q
    &:=\frac{(H_iX_t)^\top}{\sigma_q}
    \in\R^K.
    \label{eq:query_sensitivity}
\end{align}
Here $\omega_q$ is the logistic-likelihood curvature averaged over the
forecast ensemble. It is large when the threshold lies near plausible forecast
values, so that both binary responses remain possible, and small when the
response is almost predetermined. The vector $a_q$ describes how the queried
quantity changes with the ensemble coefficient $w$. In particular, the
standardized argument in \eqref{eq:query_argument_in_w} can be written as
\begin{equation}
    z_q(w)
    =
    \frac{H_i\bar{x}_t^f-\theta}{\sigma_q}
    +a_q^\top w.
    \label{eq:query_argument_affine_form}
\end{equation}

The query negative log-likelihood in
\eqref{eq:query_negative_loglikelihood} has Hessian
\begin{equation}
    \nabla_w^2\ell_q(w;o)
    =
    \sigmoid(z_q(w))\{1-\sigmoid(z_q(w))\}
    a_qa_q^\top.
    \label{eq:query_likelihood_hessian}
\end{equation}
Equation~\eqref{eq:query_likelihood_hessian} is the exact local curvature at a
specified coefficient $w$. Because $w$ is unknown when the query is selected,
we replace its scalar curvature by the forecast-ensemble average $\omega_q$ and
define
\begin{equation}
    M_q
    :=
    \omega_q a_qa_q^\top
    \succeq0.
    \label{eq:local_fisher_contribution}
\end{equation}
The matrix $M_q$ is therefore the forecast-averaged local information supplied
by query $q$ in the ensemble-coefficient space. The scalar $\omega_q$ controls
the strength of the information, while $a_qa_q^\top$ identifies the coefficient
direction informed by the query.

The local negative log-posterior is the sum of the Gaussian background penalty
and the query negative log-likelihood. Their Hessians therefore add. Since the
background term $\tfrac12 w^\top w$ has Hessian $I_K$, the local posterior
precision is approximated by
\begin{equation}
    I_K+M_q.
    \label{eq:rank_one_precision_update}
\end{equation}
Accordingly, the coefficient covariance changes from $I_K$ to
$(I_K+M_q)^{-1}$ under this local approximation.

Let
\begin{equation}
    G_t:=X_t^\top X_t.
    \label{eq:gram_matrix}
\end{equation}
Mapping the coefficient covariances through $X_t$ gives the approximate
reduction in total state variance,
\begin{equation}
    \widehat{\mathcal U}_{\mathrm{EG}}(q)
    :=
    \tr(G_t)
    -
    \tr\!\left[G_t(I_K+M_q)^{-1}\right].
    \label{eq:rank_one_trace_reduction}
\end{equation}
Using the Sherman--Morrison formula, this becomes
\begin{equation}
    \widehat{\mathcal U}_{\mathrm{EG}}(q)
    =
    \frac{\omega_q a_q^\top G_ta_q}
    {1+\omega_q a_q^\top a_q}.
    \label{eq:expected_gain_proxy}
\end{equation}
This proxy retains the state-space trace objective of the outcome-conditioned criterion but
avoids the two outcome-conditioned IDA analyses.

\paragraph{Curvature--variance proxy.}
A still lighter criterion follows from the local log-determinant information
gain of a single query. Log-determinant criteria measure the change in the
volume of a local Gaussian uncertainty ellipsoid and are commonly used in
D-optimal experimental design~\citep{chaloner1989bayesian}. Using the same
$M_q$ as in \eqref{eq:local_fisher_contribution},
\begin{align}
    \Delta_D^{(1)}(q)
    &:={}
    \log\det(I_K+M_q)
    =
    \log\!\left(1+\omega_q a_q^\top a_q\right),
    \label{eq:single_query_logdet_gain}
\end{align}
where the superscript $(1)$ indicates a single-query gain and the second
equality follows from the matrix determinant lemma. Since $\log(1+x)$ is
strictly increasing for $x\geq 0$,
\begin{equation}
    \arg\max_{q\in\cQ_t}\Delta_D^{(1)}(q)
    =
    \arg\max_{q\in\cQ_t}\omega_q a_q^\top a_q.
    \label{eq:single_query_ranking}
\end{equation}
We therefore define the curvature--variance score
\begin{equation}
    \psi_{\mathrm{cv}}(q)
    :=
    \omega_q a_q^\top a_q.
    \label{eq:proxy_score}
\end{equation}
Thus, $\psi_{\mathrm{cv}}$ gives exactly the same single-query ranking as the
local log-det gain under the rank-one information surrogate.
Since
\begin{equation}
    a_q^\top a_q
    =
    \frac{H_iP_{x,t}^fH_i^\top}{\sigma_q^2},
    \label{eq:sensitivity_variance_relation}
\end{equation}
where $H_iP_{x,t}^fH_i^\top$ is the forecast variance of the queried scalar,
the score can be written as
\begin{equation}
    \psi_{\mathrm{cv}}(q)
    =
    \frac{\omega_q}{\sigma_q^2}
    H_iP_{x,t}^fH_i^\top.
\end{equation}
Thus, the score combines the forecast uncertainty of the queried scalar with
the forecast-averaged logistic curvature associated with the candidate
threshold. It is large when substantial uncertainty remains in the queried
quantity and the threshold lies in an informative part of the forecast
distribution. The score is a local information criterion, not an exact
prediction of the reduction in nonlinear filtering error.

\subsection{Log-det greedy selection using a local Fisher surrogate}
\label{subsec:logdet_greedy}

The rank-one expected-gain and curvature--variance proxies rank candidates
individually. As a result, several high-scoring queries may still probe nearly
the same ensemble-space direction. The log-det greedy method instead evaluates
each candidate relative to the queries already selected.

For a selected query set $\cS\subseteq\cQ_t$, define the approximate local
precision
\begin{equation}
    A(\cS)
    :=
    A_0+\sum_{q\in\cS}M_q,
    \qquad
    A_0=I_K.
    \label{eq:set_precision}
\end{equation}
Here $A_0=I_K$ is the precision induced by the Gaussian background term. The
set utility is
\begin{equation}
    F_D(\cS)
    :=
    \log\det A(\cS)-\log\det A_0.
    \label{eq:logdet_surrogate}
\end{equation}
Given the current set $\cS$, the gain from adding query $q$ is
\begin{align}
    \Delta_D(q\mid\cS)
    &:={}
    F_D(\cS\cup\{q\})-F_D(\cS)
    \nonumber\\
    &=
    \log\!\left(
        1+\omega_q a_q^\top A(\cS)^{-1}a_q
    \right).
    \label{eq:logdet_marginal}
\end{align}
The quadratic form $a_q^\top A(\cS)^{-1}a_q$ measures the remaining local
uncertainty in the direction probed by $q$. Once similar directions have
already been constrained, this quantity and the marginal gain decrease.

\begin{proposition}[Submodularity of the local log-det query utility]
\label{prop:submodular_logdet}
Let $\cQ_t$ be a finite candidate query set, let $A_0\succ0$, and let each query
$q\in\cQ_t$ contribute a fixed rank-one matrix
$M_q=\omega_q a_q a_q^\top\succeq0$, with $\omega_q\ge0$, where $M_q$ is held
fixed during the greedy selection. Then
\begin{equation}
    F_D(\cS)
    =
    \log\det\!\left(A_0+\sum_{q\in\cS}M_q\right)
    -\log\det A_0
\end{equation}
is normalized, monotone nondecreasing, and submodular.

Let $\cS_{\mathrm{greedy}}$ denote the set obtained by greedily selecting,
at each step, the query with the largest marginal gain
$\Delta_D(q\mid\cS)$ until $B$ queries have been selected, and let
\begin{equation}
    \cS^\star
    \in
    \arg\max_{\cS\subseteq\cQ_t,\ |\cS|\le B}
    F_D(\cS)
\end{equation}
denote an optimal query set under the same cardinality constraint. Then
\begin{equation}
    F_D(\cS_{\mathrm{greedy}})
    \ge
    \left(1-\frac{1}{e}\right)
    F_D(\cS^\star).
\end{equation}
Thus, the greedy selection attains at least approximately $63.2\%$ of the
optimal value of the fixed surrogate objective $F_D$.
\end{proposition}

\begin{proof}
Normalization follows from $F_D(\emptyset)=0$. Because $M_q\succeq0$ for every
$q$, adding a query can only increase $A(\cS)$ in the positive-semidefinite
order and therefore cannot decrease its log determinant. Hence $F_D$ is
monotone.

For submodularity, let $\cS_1\subseteq\cS_2\subseteq\cQ_t$ and
$q\in\cQ_t\setminus\cS_2$. Then
$A(\cS_1)\preceq A(\cS_2)$, and inversion reverses this order:
\begin{equation}
    A(\cS_2)^{-1}\preceq A(\cS_1)^{-1}.
\end{equation}
It follows that
\begin{equation}
    a_q^\top A(\cS_2)^{-1}a_q
    \le
    a_q^\top A(\cS_1)^{-1}a_q.
\end{equation}
Since $\omega_q\ge0$ and $\log(1+x)$ is nondecreasing for $x\ge0$,
substitution into \eqref{eq:logdet_marginal} gives
\begin{equation}
    \Delta_D(q\mid\cS_2)
    \le
    \Delta_D(q\mid\cS_1),
\end{equation}
which is the diminishing-returns property. The bound
\begin{equation}
    F_D(\cS_{\mathrm{greedy}})
    \ge
    \left(1-\frac{1}{e}\right)
    F_D(\cS^\star)
\end{equation}
then follows from the classical result for normalized monotone submodular
maximization under a cardinality constraint~\citep{nemhauser1978}.
\end{proof}

After a query is selected, $A(\cS)^{-1}$ can be updated by the
Sherman--Morrison formula,
\begin{equation}
    A(\cS\cup\{q\})^{-1}
    =
    A(\cS)^{-1}
    -
    \frac{
        \omega_q A(\cS)^{-1}a_qa_q^\top A(\cS)^{-1}
    }{
        1+\omega_q a_q^\top A(\cS)^{-1}a_q
    }.
    \label{eq:sherman_morrison_set_update}
\end{equation}
Thus, log-det greedy accounts for redundancy without hypothetical IDA analyses
for every candidate response. It is more expensive than one-shot
curvature--variance ranking but substantially lighter than the outcome-conditioned
expected-gain criterion.

A practical no-duplicate-target variant restricts the feasible set so that at
most one threshold is selected for each target during one assimilation cycle.
This restriction encourages target diversity but changes the feasible family.
The approximation bound in Proposition~\ref{prop:submodular_logdet} therefore
applies to greedy maximization of $F_D$ under the cardinality constraint
$|\cS|\le B$, and does not directly apply to the no-duplicate-target variant
or to the final nonlinear filtering RMSE. The realized responses are still
assimilated by solving the original logistic IDA objective.

\subsection{Comparison of the joint query-selection criteria}
\label{subsec:criterion_comparison}

The four criteria differ in the uncertainty measure they approximate and in
whether they account for interactions among selected queries. The outcome-conditioned
expected-gain criterion evaluates both hypothetical responses and directly
estimates their expected state-space trace reduction, but it is the most
expensive. The rank-one expected-gain proxy uses a local Fisher update while
retaining the same trace objective. The curvature--variance proxy is the
lightest method and gives the same candidate ranking as the single-query log-det
gain under the local rank-one information surrogate. The log-det greedy method
extends the same local information model to a set objective and explicitly
discounts redundant directions.

In the experiments, the outcome-conditioned expected-gain criterion, rank-one
expected-gain proxy, curvature--variance proxy, and log-det greedy method are
denoted by AQD-EG full, AQD-EG proxy, AQD-curvature, and AQD-log-det,
respectively.

\subsection{Algorithmic summary}
\label{subsec:algorithmic_summary}

Algorithm~\ref{alg:aqdida} summarizes the four query-selection modes. The
selection rule changes across the variants, but all selected binary responses
are assimilated by the same logistic IDA update.

\begin{algorithm}[t]
\caption{Active Query Design for IDA}
\label{alg:aqdida}
\begin{algorithmic}[1]
\Require forecast ensemble $\cX_t^f$, candidate query set $\cQ_t$, budget $B$
\State compute $\bar{x}_t^f$ and $X_t$
\If{using AQD-EG full}
    \For{each $q\in\cQ_t$}
        \State solve the two outcome-conditioned hypothetical IDA problems
        \State evaluate \eqref{eq:utility}
    \EndFor
    \State select the top $B$ individual scores
\ElsIf{using AQD-EG proxy}
    \For{each $q\in\cQ_t$}
        \State compute \eqref{eq:expected_gain_proxy}
    \EndFor
    \State select the top $B$ individual scores
\ElsIf{using AQD-curvature}
    \For{each $q\in\cQ_t$}
        \State compute \eqref{eq:proxy_score}
    \EndFor
    \State select the top $B$ individual scores
\Else
    \State compute $M_q$ in \eqref{eq:local_fisher_contribution} for all $q$
    \State initialize $\cS_t\gets\emptyset$ and $A(\cS_t)^{-1}\gets A_0^{-1}$
    \For{$r=1,\ldots,B$}
        \State select the query maximizing \eqref{eq:logdet_marginal}
        \State update $\cS_t$ and $A(\cS_t)^{-1}$ using \eqref{eq:sherman_morrison_set_update}
    \EndFor
\EndIf
\State issue the selected queries and obtain responses $o_q\in\{0,1\}$
\State encode each response as a logistic inequality likelihood
\State solve the IDA objective and update the analysis ensemble
\end{algorithmic}
\end{algorithm}

\section{Experimental Evaluation}
\label{sec:experiments}

This section evaluates AQD-IDA under a fixed binary-query budget. In the experimental notation, $N_e$ denotes the ensemble size and corresponds to $K$ in the methodological development; $\sigma_o$ denotes the common logistic noise scale, so $\sigma_q=\sigma_o$ for every candidate query. The experiments are designed to address three questions: (i) whether active interval-query selection improves over passive threshold observations, (ii) how active query selection interacts with ensemble spread calibration in a closed-loop DA system, and (iii) whether the observed conclusions persist in a more structured geophysical model. Experiments were conducted on a computer running Windows 11, equipped with an Intel Core i7 processor and 32 GB of RAM.

\subsection{Lorenz--96 binary-query benchmark}
\label{subsec:lorenz96}

\paragraph{Experimental setup.}
We evaluate binary threshold observations on a perfect-model Lorenz--96 data-assimilation benchmark. The purpose of this experiment is to compare different ways of specifying binary threshold queries under the same observation budget. The observation budget is fixed to $B=12$ scalar threshold queries per assimilation cycle. The methods differ in how the query pairs $q=(i,\theta)$ are specified. The fixed-threshold reference uses prespecified coordinates and constant thresholds. The adaptive-threshold variants keep the coordinates fixed but choose the thresholds from the forecast ensemble. The AQD methods choose both the queried coordinates and thresholds adaptively.

The Lorenz--96 state dimension is $n=40$, with forcing $F=8$. The continuous-time dynamics are
\begin{equation}
    \frac{d x_i}{dt}
    = (x_{i+1}-x_{i-2})x_{i-1}-x_i+F,
    \qquad i=1,\ldots,n,
\end{equation}
with cyclic indexing $x_{i+n}=x_i$. The model is advanced with time step $\Delta t=0.05$. For each seed, the truth is spun up for 1460 model steps, and the initial ensemble is generated by adding independent Gaussian perturbations with standard deviation 1.0 to the spun-up truth. The ensemble size is $N_e=16$. Each run consists of 14600 assimilation cycles, and all reported RMSE and spread statistics are averages over the last half of the run, i.e., the final 7300 cycles.

Once a query pair $q=(i,\theta)$ has been specified, the binary observation is generated in the same way for all methods:
\begin{equation}
    o = \mathbf{1}\{x_i + \varepsilon \geq \theta\},
    \qquad
    \varepsilon \sim \mathrm{Logistic}(0,\sigma_o),
\end{equation}
with logistic noise scale $\sigma_o=1.0$. The same interval data assimilation update is then applied to all methods. Thus, the compared methods differ in how the binary threshold queries are specified, not in the observation model or the assimilation update.

Multiplicative inflation is applied after each assimilation update to adjust
the ensemble spread. Since the analysis accuracy can be sensitive to the
inflation factor, we evaluate all methods over multiple inflation values,
\begin{equation}
    \alpha \in \{1.00,1.04,1.08,1.10,1.12,1.14,1.16\}.
\end{equation}
Here, $\alpha=1.00$ corresponds to no multiplicative inflation, whereas larger
values increase the analysis ensemble spread after the update. This sweep is
used to examine how each query-specification rule behaves under different
spread-correction levels. 

\paragraph{Candidate queries.}
Let
\[
    \mathcal{I}_{\mathrm{fix}}
    =
    \{1,4,8,11,15,18,22,25,29,32,36,40\}
\]
denote the fixed coordinate set. For the fixed-threshold reference, we use these fixed coordinates and evaluate constant thresholds from
\[
    \Theta_{\mathrm{fix}}=\{-4,-2,0,1,2,4,6\}.
\]
For each value of $\theta\in\Theta_{\mathrm{fix}}$, the same threshold is used at all fixed coordinates and at all assimilation cycles.

For the methods that select thresholds from a discrete forecast-dependent candidate set, the candidate thresholds are constructed from the forecast ensemble at each cycle. Specifically, for coordinate $i$, we use
\[
    \Theta_t(i)
    =
    \left\{
    Q_{0.10}(x^f_{t,i}),
    Q_{0.25}(x^f_{t,i}),
    Q_{0.50}(x^f_{t,i}),
    Q_{0.75}(x^f_{t,i}),
    Q_{0.90}(x^f_{t,i})
    \right\}
    \cup \{1.0\},
\]
where $Q_p(x^f_{t,i})$ denotes the $p$-quantile of the forecast ensemble at coordinate $i$. The additional absolute threshold $1.0$ is included so that the adaptive rules based on $\Theta_t(i)$ can also select a fixed reference threshold.

The fixed-location adaptive-threshold variants use the same coordinate set $\mathcal{I}_{\mathrm{fix}}$. The random-threshold rule samples the threshold uniformly from the interval between the 10th and 90th percentiles of the forecast ensemble at each fixed observation coordinate. The median-threshold rule selects $Q_{0.50}(x^f_{t,i})$, while the curvature-threshold rule selects the candidate threshold in $\Theta_t(i)$ with the largest local curvature-based score. These variants isolate the effect of adapting thresholds while keeping the observation locations unchanged.

The AQD methods use the candidate query set
\[
    \mathcal{Q}_t
=
\{(i,\theta): i\in\{1,\ldots,n\},\ \theta\in\Theta_t(i)\},
\]
and select $B$ query pairs from this set at each assimilation cycle. Thus, the
AQD methods adapt both the queried coordinates and the thresholds. For AQD-curvature, AQD-EG proxy, and AQD-EG full, the Lorenz–96 experiments use a spatially diversified selection rule: only the highest-scoring threshold is retained at each coordinate, and immediately adjacent coordinates are excluded during selection. AQD-log-det instead uses the set-based greedy criterion in Section 3.

For compactness, the experimental tables use the following method names.
``Fixed-threshold'' denotes the fixed-location, constant-threshold reference,
and ``Median-threshold'' denotes fixed-location threshold adaptation using the
forecast median. Among the AQD methods, ``AQD-curvature'' denotes the
curvature--variance proxy, ``AQD-log-det'' denotes the log-det greedy
criterion, and ``AQD-EG proxy'' denotes the lightweight expected-gain proxy.
We also include ``AQD-log-det-no-dup'', a no-duplicate variant of AQD-log-det
that prevents selecting the same coordinate more than once in a cycle. The outcome-conditioned expected-gain criterion is labeled
``AQD-EG full''.

\paragraph{Performance metrics.}
We report the analysis RMSE,
\begin{equation}
    \mathrm{RMSE}_t = \sqrt{\frac{1}{n}\|\bar{x}^a_t-x_t\|_2^2},
\end{equation}
and the analysis ensemble spread,
\begin{equation}
    \mathrm{Spread}_t = \sqrt{\frac{1}{n}\sum_{i=1}^n \mathrm{Var}_{k=1}^{N_e}(x^{a,(k)}_{t,i})}.
\end{equation}
The ratio spread/RMSE is used as a simple calibration diagnostic. A run is classified as divergent if it produces non-finite analysis RMSE values at any cycle or if the maximum analysis RMSE over the full run exceeds 100. Unless otherwise stated, means are computed over the last 7300 cycles of the non-divergent runs, and the number of non-divergent seeds is reported in the tables.

\paragraph{Results.}
Tables~\ref{tab:l96_rmse_inflation} and~\ref{tab:l96_spread_inflation} report the RMSE and spread/RMSE values over the tested inflation values. The fixed-threshold row uses $\theta=1.0$, which is the best fixed-threshold candidate identified in Table~\ref{tab:l96_fixed_threshold_sweep}. At low inflation, the AQD methods are strongly underdispersed. Increasing inflation improves their spread/RMSE ratios and substantially reduces their RMSE. The fixed-location adaptive-threshold variants also improve on the fixed-threshold reference; for example, the median-threshold rule reaches RMSE $3.450$ at $\alpha=1.10$. The largest improvement is obtained when both the queried locations and thresholds are adapted.

\begin{table}[H]
\centering
\caption{Analysis RMSE over the common set of tested multiplicative inflation factors. Lower is better. Subscripts indicate the number of non-divergent runs when fewer than 10 seeds remained after applying the divergence criterion.}
\label{tab:l96_rmse_inflation}
\resizebox{\textwidth}{!}{%
\begin{tabular}{lccccccc}
\toprule
Method & $\alpha=1.00$ & $1.04$ & $1.08$ & $1.10$ & $1.12$ & $1.14$ & $1.16$ \\
\midrule
Fixed threshold & 4.219 & 3.828 & \textbf{3.732} & 3.742 & 3.786$_{8/10}$ & 3.865$_{9/10}$ & 3.971 \\
Random threshold & 4.447 & 3.865 & 3.555 & \textbf{3.522} & 3.602 & 3.703 & 3.862 \\
Median threshold & 4.492 & 3.902 & 3.511 & \textbf{3.450} & 3.470 & 3.528 & 3.619 \\
Curvature threshold & 4.473 & 3.886 & 3.559 & \textbf{3.460} & 3.515 & 3.560 & 3.695$_{9/10}$ \\
\midrule
AQD-curvature & 4.655 & 4.093 & 3.193 & 2.791 & 2.403 & 2.228 & \textbf{2.128} \\
AQD-log-det & 4.756 & 4.265 & 3.423 & 2.970 & 2.602 & 2.413 & \textbf{2.305} \\
AQD-log-det-no-dup & 4.681 & 4.144 & 3.324 & 2.861 & 2.601 & 2.421 & \textbf{2.232} \\
AQD-EG proxy & 4.664 & 4.044 & 3.047 & 2.606 & 2.218 & 2.152 & \textbf{2.110} \\
\bottomrule
\end{tabular}}
\end{table}

\begin{table}[H]
\centering
\caption{Analysis spread/RMSE over the same tested multiplicative inflation factors. Values much smaller than one indicate underdispersed ensembles.}
\label{tab:l96_spread_inflation}
\resizebox{\textwidth}{!}{%
\begin{tabular}{lccccccc}
\toprule
Method & $\alpha=1.00$ & $1.04$ & $1.08$ & $1.10$ & $1.12$ & $1.14$ & $1.16$ \\
\midrule
Fixed threshold & 0.301 & 0.492 & 0.667 & 0.746 & 0.822$_{8/10}$ & 0.892$_{9/10}$ & 0.956 \\
Random threshold & 0.177 & 0.323 & 0.500 & 0.589 & 0.664 & 0.744 & 0.821 \\
Median threshold & 0.170 & 0.301 & 0.464 & 0.542 & 0.609 & 0.670 & 0.730 \\
Curvature threshold & 0.171 & 0.303 & 0.463 & 0.552 & 0.620 & 0.695 & 0.757$_{9/10}$ \\
\midrule
AQD-curvature & 0.099 & 0.175 & 0.326 & 0.430 & 0.559 & 0.656 & 0.735 \\
AQD-log-det & 0.088 & 0.166 & 0.296 & 0.405 & 0.514 & 0.598 & 0.675 \\
AQD-log-det-no-dup & 0.095 & 0.168 & 0.302 & 0.408 & 0.500 & 0.586 & 0.690 \\
AQD-EG proxy & 0.099 & 0.179 & 0.352 & 0.477 & 0.616 & 0.687 & 0.759 \\
\bottomrule
\end{tabular}}
\end{table}

\begin{figure}[H]
\centering
\includegraphics[width=0.98\textwidth]{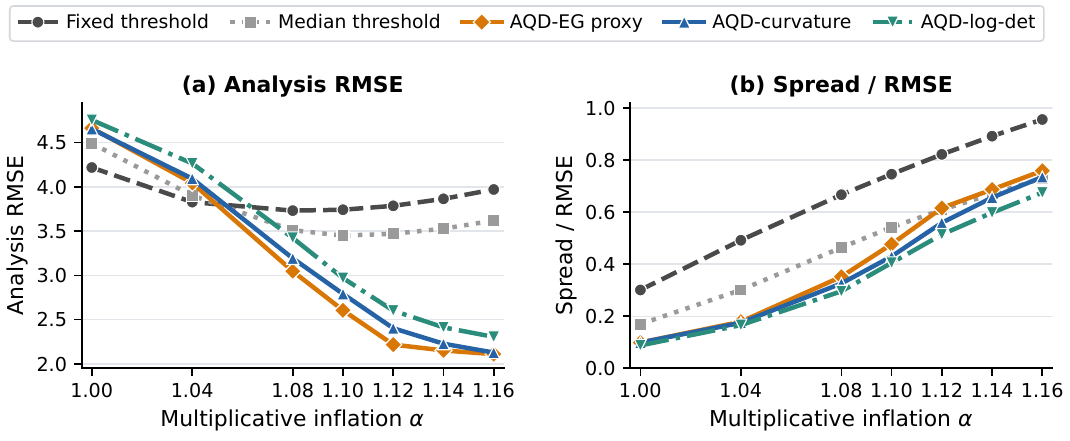}
\caption{Performance over the tested multiplicative inflation values for a representative subset of the methods. Panel (a) shows analysis RMSE and panel (b) shows the spread/RMSE ratio. AQD has the smallest spread/RMSE at low inflation, indicating stronger ensemble contraction; after inflation is increased, it achieves the lowest RMSE among the lightweight query-specification rules.}
\label{fig:l96_performance_alpha}
\end{figure}

The fixed-location adaptive-threshold variants are useful in this setting. At $\alpha=1.10$, the median-threshold rule achieves RMSE $3.450$, improving over the fixed-threshold reference with RMSE $3.742$. This indicates that threshold adaptation alone can improve the informativeness of binary observations. AQD expected-gain reaches RMSE $2.606$ at the same inflation value and $2.110$ at $\alpha=1.16$, showing the additional benefit of adapting the queried locations together with the thresholds.

\paragraph{Fixed-threshold candidates.}
For the fixed-threshold reference, the constant threshold is itself a design choice. We therefore evaluated the same fixed observation indices for the candidate thresholds
\begin{equation}
    \theta \in \{-4,-2,0,1,2,4,6\}
\end{equation}
and for the same multiplicative inflation values used above. Table~\ref{tab:l96_fixed_threshold_sweep} reports, for each threshold candidate, the best result over the tested inflation values.

\begin{table}[H]
\centering
\caption{Fixed-threshold results for the tested constant-threshold candidates. For each $\theta$, the table reports the best analysis RMSE over the tested multiplicative inflation values.}
\label{tab:l96_fixed_threshold_sweep}
\begin{tabular}{ccccc}
\toprule
Fixed threshold $\theta$ & Best $\alpha$ & Analysis RMSE & Spread/RMSE & Non-divergent runs \\
\midrule
$-4$ & 1.04 & 3.966 & 0.941 & 10/10 \\
$-2$ & 1.04 & 3.869 & 0.681 & 10/10 \\
$0$  & 1.08 & 3.761 & 0.709 & 9/10 \\
$1$  & 1.08 & \textbf{3.732} & 0.667 & 10/10 \\
$2$  & 1.08 & 3.755 & 0.647 & 10/10 \\
$4$  & 1.08 & 3.868 & 0.677 & 10/10 \\
$6$  & 1.04 & 4.009 & 0.629 & 10/10 \\
\bottomrule
\end{tabular}
\end{table}

Among the tested threshold candidates, the best fixed-threshold result is obtained at $\theta=1.0$ and $\alpha=1.08$, with an analysis RMSE of $3.732$. 

\paragraph{Expected-gain proxy and outcome-conditioned criterion.}
The expected-gain proxy is designed to approximate the uncertainty reduction of a binary query without solving an outcome-conditioned IDA problem for every possible outcome of every candidate. To check how much accuracy is lost by this approximation, we also evaluate the outcome-conditioned expected-gain criterion (AQD-EG full) for
\begin{equation}
    \alpha \in \{1.12,1.14,1.16\}.
\end{equation}
For each candidate query, AQD-EG full solves the two hypothetical IDA updates corresponding to the two possible binary outcomes and evaluates the expected posterior uncertainty. This is closer to the literal expected posterior uncertainty-reduction objective, but it is substantially more expensive.

Table~\ref{tab:l96_full_eg} compares the proxy and outcome-conditioned expected-gain criteria at the same inflation values. The two criteria give very similar RMSE values. At $\alpha=1.12$ and $\alpha=1.14$, the proxy is slightly better, while at $\alpha=1.16$ AQD-EG full is slightly better. The best AQD-EG full result is RMSE $2.097$, compared with RMSE $2.110$ for the best proxy result. This improvement is only about $0.6\%$, whereas the average query-selection time increases from about $9.0$ ms to $74.0$ ms per assimilation cycle. Thus, the proxy captures almost all of the benefit of the outcome-conditioned calculation at roughly one eighth of the selection cost.

\begin{table}[H]
\centering
\caption{Comparison between the lightweight expected-gain proxy and the outcome-conditioned expected-gain criterion (AQD-EG full). Selection time is the average query-selection time per assimilation cycle, computed over the last 7300 cycles. All entries are stable for all 10 seeds.}
\label{tab:l96_full_eg}
\resizebox{\textwidth}{!}{%
\begin{tabular}{lccccc}
\toprule
Method & $\alpha$ & Analysis RMSE & Spread/RMSE & Selection time [ms] & Non-divergent runs \\
\midrule
AQD-EG proxy & 1.12 & 2.218 & 0.616 & 8.96 & 10/10 \\
AQD-EG full  & 1.12 & 2.271 & 0.600 & 71.21 & 10/10 \\
\midrule
AQD-EG proxy & 1.14 & 2.152 & 0.687 & 8.97 & 10/10 \\
AQD-EG full  & 1.14 & 2.199 & 0.670 & 72.66 & 10/10 \\
\midrule
AQD-EG proxy & 1.16 & 2.110 & 0.759 & 8.99 & 10/10 \\
AQD-EG full  & 1.16 & \textbf{2.097} & 0.753 & 73.97 & 10/10 \\
\bottomrule
\end{tabular}}
\end{table}

\subsection{Spherical Quasi-Geostrophic (SHQG) benchmark}
\label{subsec:shqg}

\paragraph{Experimental setup.}
We next evaluate AQD-IDA on the SHQG benchmark. This experiment is intended to test whether the conclusions from the lower-dimensional Lorenz--96 study persist in a more structured geophysical model with spatially organized wind fields. The SHQG model is run with spectral resolution $T=21$, physical grid parameter $T_{\mathrm{grid}}=31$, three vertical levels, time step $\Delta t=6$ hours, and a data-assimilation window of one model step. Each run uses a one-year spin-up period (1460 model steps), followed by 3650 assimilation cycles, and the reported metrics are averaged over the last half of the run, i.e., the final 1825 cycles. All SHQG results in this section use five random seeds, with seeds $0,\ldots,4$.

The evaluation metric is the physical-wind analysis RMSE computed from the horizontal wind components $u$ and $v$. The physical wind is diagnosed from the streamfunction $\psi$ using the spherical-gradient operators associated with the spherical-harmonic representation. The binary threshold queries are applied to the physical wind components at selected vertical levels and latitude--longitude grid points. Unless otherwise stated, we use ensemble size $N_e=50$, query budget $B=128$, logistic noise scale $\sigma_o=2.2053$, and multiplicative inflation $\alpha=1.02$. Candidate observation targets are constructed on a strided latitude--longitude grid with stride 4 in both directions, and the number of candidate queries is capped at 2000 per cycle. For adaptive-threshold and AQD rules, the threshold candidates at each candidate target are the forecast quantiles $25\%$, $50\%$, and $75\%$.

For the fixed-threshold reference, the constant threshold is varied over
\(\theta\in\{-4,-2,0,2,4\}\) under the representative SHQG configuration, and
the best resulting fixed-threshold entry is used in the main comparison. The
Median-threshold and AQD entries follow the method definitions introduced in
the Lorenz--96 benchmark.

\paragraph{Main comparison.}
Table~\ref{tab:shqg_main_results} reports the main SHQG comparison together with the fixed-threshold sweep. The
fixed-threshold entry uses $\theta=2$, the best constant threshold in the tested set. The median-threshold rule improves
substantially over this fixed-threshold reference, reducing the analysis RMSE
from \(0.9255\) to \(0.4178\). Joint location--threshold selection gives a
further reduction. AQD-curvature achieves the lowest mean RMSE, \(0.1937\),
while AQD-log-det gives a very similar value, \(0.1945\). These values
correspond to an approximately \(79.1\%\) reduction relative to the best
fixed-threshold reference and a \(53.6\%\) reduction relative to the
median-threshold rule. AQD-EG proxy also improves over both references,
although it is less effective than AQD-curvature and AQD-log-det in this
configuration.

\begin{table}[H]
\centering
\caption{SHQG representative comparison and fixed-threshold sweep over five seeds. Values are last-half physical-wind analysis RMSE. Lower RMSE is better.}
\label{tab:shqg_main_results}
\small
\textbf{(a) Main comparison}\par\vspace{0.25em}
\begin{tabular}{lcccc}
\toprule
Method & Analysis RMSE & Std. & Selection time [ms] & Unique locations \\
\midrule
Fixed-threshold & 0.9255 & 0.0198 & -- & 128.0 \\
Median-threshold & 0.4178 & 0.0083 & 23.14 & 128.0 \\
AQD-EG proxy & 0.2669 & 0.0041 & 29.43 & 49.8 \\
AQD-curvature & \textbf{0.1937} & 0.0061 & 27.18 & 50.0 \\
AQD-log-det & 0.1945 & 0.0036 & 79.23 & 51.9 \\
\bottomrule
\end{tabular}

\vspace{0.8em}
\textbf{(b) Fixed-threshold sweep}\par\vspace{0.25em}
\begin{tabular}{ccc}
\toprule
Fixed threshold $\theta$ & Analysis RMSE & Std. \\
\midrule
$-4$ & 1.2061 & 0.0209 \\
$-2$ & 1.0489 & 0.0316 \\
$0$  & 0.9733 & 0.0283 \\
$2$  & \textbf{0.9255} & 0.0198 \\
$4$  & 0.9430 & 0.0165 \\
\bottomrule
\end{tabular}
\end{table}

Figure~\ref{fig:shqg_main} shows the same main comparison visually. The error bars
are much smaller than the gap between the fixed-threshold reference, the
median-threshold rule, and the best AQD methods, indicating that the observed
differences are not explained by seed-to-seed variability.
\begin{figure}[H]
\centering
\includegraphics[width=0.86\textwidth]{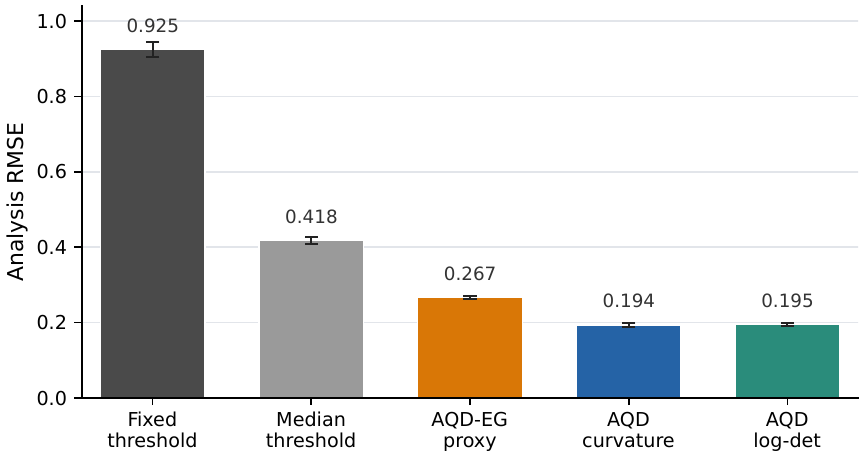}
\caption{SHQG main comparison over five seeds. Bars show last-half physical-wind analysis RMSE with one-standard-deviation error bars.}
\label{fig:shqg_main}
\end{figure}

\paragraph{Fixed-threshold sweep.}
The lower block of Table~\ref{tab:shqg_main_results} reports the fixed-threshold sweep
used to define the fixed-threshold reference in the main comparison. Among the
tested constant thresholds, $\theta=2$ gives the lowest RMSE, $0.9255$. This
value is used as the fixed-threshold entry in the main SHQG comparison.

\subsection{SHQG ablations}
\label{subsec:shqg_ablations}

We next vary the query budget, ensemble size, and logistic noise scale one
factor at a time, using the representative SHQG setting as the reference:
\(B=128\), \(N_e=50\), and \(\sigma_o=2.2053\). In each ablation, all other
parameters are kept fixed at their representative values. For the Fixed-threshold
baseline, \(\theta=2\) is used throughout all ablations.

\begin{table}[H]
\centering
\caption{SHQG ablations over five seeds. The upper table reports the last-half
physical-wind analysis RMSE. The lower table reports the average query-selection
time per assimilation cycle for the budget ablation. In each ablation, only the
indicated factor is varied from the representative setting
$B=128$, $N_e=50$, and $\sigma_o=2.2053$.}
\label{tab:shqg_ablations}
\scriptsize

\begin{tabular}{llccccc}
\toprule
Factor & Setting
& Fixed-threshold
& Median-threshold
& AQD-EG proxy
& AQD-curvature
& AQD-log-det \\
\midrule
Budget $B$
& 32  & unstable$^{\dagger}$ & 1.198 & 0.557 & 0.391 & \textbf{0.383} \\
& 64  & 1.406 & 0.613 & 0.382 & \textbf{0.266} & 0.267 \\
& 128 & 0.925 & 0.418 & 0.267 & \textbf{0.194} & 0.195 \\
& 256 & 0.505 & 0.269 & 0.193 & 0.147 & \textbf{0.146} \\
\midrule
Ensemble size $N_e$
& 25  & 0.915 & 0.368 & 0.208 & \textbf{0.139} & $0.677^{\ddagger}$ \\
& 50  & 0.925 & 0.418 & 0.267 & \textbf{0.194} & 0.195 \\
& 100 & 1.057 & 0.478 & 0.319 & 0.241 & \textbf{0.241} \\
\midrule
Logistic scale $\sigma_o$
& 1.1027 & 0.629 & 0.203 & 0.137 & 0.097 & \textbf{0.097} \\
& 2.2053 & 0.925 & 0.418 & 0.267 & \textbf{0.194} & 0.195 \\
& 4.4106 & 1.493 & 0.919 & 0.548 & 0.402 & \textbf{0.401} \\
\bottomrule
\end{tabular}

\vspace{0.8em}

\begin{tabular}{cccc}
\toprule
& \multicolumn{3}{c}{Query-selection time [ms]} \\
\cmidrule(lr){2-4}
Budget $B$
& AQD-EG proxy
& AQD-curvature
& AQD-log-det \\
\midrule
32  & 30.1 & 27.2 & 41.8 \\
64  & 30.6 & 27.9 & 53.8 \\
128 & 29.4 & 27.2 & 79.2 \\
256 & 30.2 & 28.9 & 179.6 \\
\bottomrule
\end{tabular}

\vspace{0.4em}
\begin{minipage}{0.94\textwidth}
\footnotesize
$^{\dagger}$At $B=32$, one fixed-threshold seed produced a non-finite RMSE
and three finite runs exhibited severe instability, so no aggregate RMSE is
reported.\\
$^{\ddagger}$For $N_e=25$, AQD-log-det had one unstable seed with analysis
RMSE $2.733$; the other four seeds were much closer to the AQD-curvature values.
\end{minipage}
\end{table}

The query-budget part of Table~\ref{tab:shqg_ablations} shows a consistent
budget trend. Increasing $B$ improves all stable methods, and the ordering is
stable: the median-threshold rule improves over the fixed-threshold reference,
while AQD-curvature and AQD-log-det provide the lowest RMSE across all tested
budgets.

The lower part of Table~\ref{tab:shqg_ablations} reports the query-selection
times. AQD-log-det becomes more expensive as the query budget increases,
rising from $41.8$ ms at $B=32$ to $179.6$ ms at $B=256$. By contrast,
AQD-curvature remains nearly constant, between $27.2$ and $28.9$ ms, while
achieving RMSE values close to those of AQD-log-det. AQD-EG proxy also remains
around $30$ ms, but its RMSE is consistently higher than those of curvature
and log-det in this SHQG setting.

The ensemble-size ablation shows that the advantage of AQD is maintained
across the tested values of $N_e$. Increasing the ensemble size from
$N_e=50$ to $N_e=100$ does not substantially change the overall performance
of the AQD methods. AQD-curvature consistently achieves the lowest or
near-lowest RMSE, while AQD-log-det also performs well except for one
unstable seed at $N_e=25$.

The logistic-scale part of Table~\ref{tab:shqg_ablations} shows the expected
degradation as the logistic scale increases: sharper binary likelihoods provide
stronger information, whereas larger scales make each binary response less
informative. The relative ordering is nevertheless stable. Median-threshold
selection improves over the fixed-threshold reference, and
AQD-curvature/AQD-log-det remain the best methods at all tested noise scales.

\begin{figure}[H]
\centering
\includegraphics[width=0.98\textwidth]{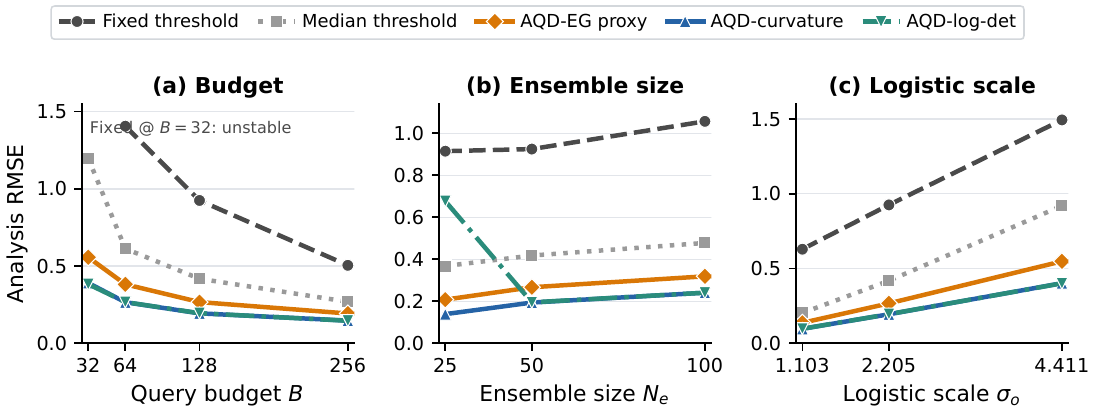}
\caption{SHQG ablation trends over five seeds. Curves show mean analysis RMSE
across seeds for (a) query budget, (b) ensemble size, and (c) logistic scale.
The fixed-threshold point at $B=32$ is omitted because that case is unstable.}
\label{fig:shqg_ablations}
\end{figure}

\section{Related Work and Positioning}
\label{sec:related-work}

This section reviews related work on data assimilation with nonstandard observations, adaptive observation and data selection, and active sensing and experimental design.

\subsection{Ensemble data assimilation and non-Gaussian ensemble updates}

Ensemble data assimilation represents forecast uncertainty using an ensemble of model realizations. 
The ensemble Kalman filter (EnKF) introduced Monte Carlo estimates of forecast-error statistics into sequential geophysical data assimilation \citep{evensen1994sequential}, while deterministic square-root formulations such as the local ensemble transform Kalman filter (LETKF) enabled efficient ensemble updates in high-dimensional spatiotemporal systems \citep{hunt2007letkf}. 
These methods provide the forecast ensemble, ensemble-space perturbations, and sample covariance information used throughout this paper.

The classical EnKF framework is most directly suited to settings in which forecast and observation uncertainties are reasonably approximated by Gaussian distributions. 
A substantial body of work has therefore considered ensemble-based updates for nonlinear and non-Gaussian problems. 
Particle filtering provides a general Bayesian approach without imposing a Gaussian posterior approximation, although direct application to high-dimensional geophysical systems is challenging because of particle degeneracy \citep{vanleeuwen2009particle}. 
Localized particle-filter formulations have subsequently been developed to alleviate this difficulty in large dynamical systems \citep{poterjoy2016localized}. 
Within ensemble-filtering approaches, the rank histogram filter performs scalar Bayesian updates using ensemble-based representations of non-Gaussian prior distributions \citep{anderson2010nongaussian}, while moment-matching and Gaussian-anamorphosis approaches provide alternative mechanisms for handling nonlinear or non-Gaussian observation--state relationships \citep{lei2011moment,amezcua2014gaussian}.

More recently, the quantile-conserving ensemble filter framework (QCEFF) has provided a flexible scalar-update formulation in which prior densities and likelihoods can be non-Gaussian and the posterior ensemble is constructed through quantile-preserving transformations \citep{anderson2022qceff1}. 
Subsequent developments use probit and probability-integral transformations to regress observation-space increments back to state variables while accommodating bounded and strongly non-Gaussian variables \citep{anderson2023qceff2}.

AQD-IDA is complementary to these developments. 
Rather than introducing a general-purpose non-Gaussian ensemble filter, it focuses on interval, inequality, and binary-threshold observations and addresses the preceding observation-design problem: which such observations should be actively acquired before they are assimilated through the IDA update.

\subsection{Inequality, censored, qualitative, binary, and interval observations}

Several data-assimilation studies have recognized that useful information may be inequality-valued rather than point-valued. Inequality constraints have been incorporated into ocean data assimilation to enforce or exploit physical restrictions such as nonnegative layer thicknesses \citep{thacker2007inequality}. Related constrained-filtering approaches, including truncated-Gaussian ideas, treat bounds as part of the posterior structure rather than as ordinary Gaussian observations \citep{lauvernet2009truncated}. These works show that inequality information can be meaningful in geophysical estimation, but the inequalities are typically prescribed constraints or fixed observation limits.

Semi-qualitative ensemble filtering provides another close connection. The EnKF-SQ was developed for observations with detection limits, where out-of-range values are qualitative by nature and can be interpreted as inequalities rather than discarded \citep{shah2018semiqualitative}. It has also been applied in a coupled ocean--sea-ice setting for sea-ice-thickness observations with upper detection limits \citep{shah2019seaice}. More recently, binary data have been analyzed in variational data assimilation from the viewpoints of information worth, observability, and identifiability \citep{srivastava2025binary}. These studies are highly relevant to AQD-IDA because they establish that coarse, range-limited, or binary observations can carry useful state information.

Interval Data Assimilation (IDA) is the direct methodological foundation of the present paper. IDA formulates equality, inequality, and finite-interval observations within a single Bayesian framework by treating all observations as intervals and by using smooth logistic interval likelihoods \citep{leduc2026ida}. In IDA, a point observation is a zero-width interval, a one-sided inequality is a half-infinite interval, and a finite interval is represented by lower and upper bounds. AQD-IDA keeps this observation philosophy but changes the problem setting: the interval observation is not assumed to be given. Instead, the sensing system chooses a query, such as ``is $H_i x_t > \theta$?'', and the location--threshold pair $(i,\theta)$ becomes the object of active design.

\subsection{Targeted observations, observation impact, and data selection}

Adaptive observing is a long-standing theme in meteorology. \citet{lorenz1998optimal} studied optimal sites for supplementary observations in a 40-variable model, showing how additional observations can be targeted to improve forecasts. \citet{bishop2001adaptive} developed adaptive sampling with the ensemble transform Kalman filter, using ensemble information to estimate how candidate observations would reduce forecast error variance. The broader targeted-observation literature, reviewed by \citet{majumdar2016review}, emphasizes that targeted observations can be beneficial but that their impact depends strongly on the verification metric, observing platform, numerical model, data-assimilation system, and case selection.

A related operational literature evaluates the impact of observations after they are assimilated. Adjoint-based forecast-sensitivity-to-observations methods estimate how each assimilated observation changes a chosen forecast-error measure \citep{langland2004impact}. Ensemble-based diagnostics and EFSO-type methods have also been used to accelerate the development of selection criteria for new observing systems; for example, \citet{lien2018efso} used EFSO statistics to refine data-selection criteria for satellite precipitation assimilation. Information-based data selection for ensemble assimilation similarly selects a subset of already available observations according to expected information content and finite-ensemble considerations \citep{migliorini2013information}.

AQD-IDA shares the flow-dependent nature of targeted observing but addresses a different observation-design problem. First, rather than selecting only where or when observations should be acquired, AQD-IDA determines what binary query to issue by jointly selecting the observation target and inequality threshold. Second, its query-selection criteria are derived directly from the logistic likelihood and posterior uncertainty structure used in the subsequent IDA update. Third, whereas targeted observing typically seeks observations that improve a subsequent forecast, AQD-IDA designs the information supplied to the current assimilation update.

\subsection{Active learning, binary-response design, and one-bit sensing}

Outside geophysical data assimilation, AQD-IDA is conceptually related to active learning and Bayesian experimental design. Active learning selects examples or queries whose responses are expected to be informative for a predictive model \citep{settles2009active,cacciarelli2024active}, while Bayesian experimental design chooses experimental conditions to reduce posterior uncertainty, including sequential settings in which later designs depend on previously acquired data \citep{chaloner1989bayesian,foster2021deep}. These approaches share with AQD-IDA the general principle that the information obtained from an observation depends on the query or experimental condition that is chosen.

One-bit and quantized sensing provide a more direct connection to binary-threshold observations. In one-bit sensing, measurements retain only whether a linear quantity lies above or below a threshold, and adaptive methods can adjust these thresholds using previous estimates to improve reconstruction \citep{fang2016adaptive,beheshti2022adaptive}. Recent work has also examined how nonzero or time-varying thresholds affect Fisher information and estimation accuracy in one-bit estimation \citep{xiao2023covariance}. These studies establish that the threshold of a binary measurement can itself be an important design variable.

AQD-IDA, however, addresses a different estimation problem. Rather than reconstructing a fixed signal or estimating static parameters, it performs sequential data assimilation for a dynamically evolving state. At each cycle, the forecast ensemble is used to determine which quantity to query and what inequality to ask, and the resulting binary response is assimilated through IDA before propagation to the next cycle. The selection criteria are derived from the uncertainty and logistic-likelihood structure of the subsequent IDA update.

\section{Conclusion}\label{sec:conclusion}

This paper proposed Active Query Design for Interval Data Assimilation
(AQD-IDA), a framework for selecting informative binary-threshold observations
before they are acquired and assimilated through IDA. Rather than treating
observation locations and thresholds as fixed, AQD-IDA uses the
flow-dependent forecast ensemble to jointly design the observation target
and threshold under a finite query budget. We developed several
query-selection criteria based on likelihood curvature, local information,
and expected posterior uncertainty reduction, including computationally
efficient rank-one and log-det formulations in the ensemble coefficient
space.

Experiments on Lorenz--96 and SHQG demonstrated that adaptive threshold
selection can substantially improve assimilation accuracy over fixed
thresholds, and that jointly selecting observation targets and thresholds
provides further gains under the same binary-query budget. The results also
highlight the interaction between active query design and ensemble spread
calibration. Among the proposed criteria, the curvature--variance score
provided a particularly favorable balance between assimilation accuracy
and computational cost, while the log-det formulation offered a principled
way to account for redundancy among multiple queries.

These results suggest that the design of the observation itself can be an
important component of data assimilation when measurements are available
only in interval or binary form. Future work will consider adaptive spread
calibration, heterogeneous observation costs, more realistic sensing and
communication constraints, and extensions to operational observing systems.

\section*{Data Availability Statement}
The source code used for the AQD-IDA experiments is publicly available at
\url{https://github.com/hashimo262/AQD-IDA}. The repository contains the
AQD-IDA-specific code for the Lorenz--96 and SHQG experiments. The dynamical
models and associated utilities used by these experiments are provided by the
DAJax project, available at
\url{https://github.com/leducvn/dajax}.

\bibliographystyle{plainnat}
\bibliography{references}

\end{document}